%% file: main.tex
\documentclass[a4paper]{article}
\usepackage[margin=1in]{geometry}
\usepackage{enumerate}
\usepackage{amssymb}
\usepackage{amsthm}
\usepackage{amsmath}
\usepackage{graphicx}
\usepackage{stmaryrd}

\usepackage[algoruled,linesnumbered,algo2e,vlined]{algorithm2e}
\SetArgSty{textrm}

\usepackage[final,mode=multiuser]{fixme}

\input{macros}

\newtheorem{theorem}{Theorem}
\newtheorem{definition}[theorem]{Definition}
\newtheorem{lemma}[theorem]{Lemma}

\newtheorem{problem}[theorem]{Problem}
\newtheorem{fact}[theorem]{Fact}
\newtheorem{example}[theorem]{Example}

\title{Dynamic index and LZ factorization in compressed space}

 \author{
   Takaaki Nishimoto$^1$\quad
   Tomohiro I$^2$\quad
   Shunsuke Inenaga$^1$\\
   Hideo Bannai$^1$\quad
   Masayuki Takeda$^1$\\
  {$^1$ Department of Informatics, Kyushu University}\\
  {\texttt{\{takaaki.nishimoto,inenaga,bannai,takeda\}@inf.kyushu-u.ac.jp}}\\
  {$^2$ Kyushu Institute of Technology, Japan}\\
  {\texttt{tomohiro@ai.kyutech.ac.jp}}
 }

\date{}
\begin{document}
\maketitle

\input{abstract}
\input{introduction}
\input{preliminaries}
\input{framework}
\input{update}
\input{static_index}

\input{lz}

\vspace*{1pc}

\noindent \textbf{Acknowledgments.} We would like to thank Pawe{\l} Gawrychowski for drawing our attention to
the work by Alstrup et al.~\cite{LongAlstrup,DBLP:conf/soda/AlstrupBR00}
and for fruitful discussions.

\bibliographystyle{siam}
\bibliography{ref}

\clearpage
\appendix
\input{examples_figs}

\end{document}

%% file: macros.tex
\fxuseenvlayout{colorsig}
\fxusetargetlayout{color}
\FXRegisterAuthor{tn}{atn}{TN}
\FXRegisterAuthor{ti}{ati}{TI}
\FXRegisterAuthor{si}{asi}{SI}
\FXRegisterAuthor{hb}{ahb}{HB}
\FXRegisterAuthor{mt}{amt}{MT}
\fxsetface{env}{}

\newcommand{\Occ}{\mathit{Occ}}

\newcommand{\vOcc}{\mathit{vOcc}}

\newlength\savedwidth

\newcommand{\rev}[1]{#1^{R}}

\newcommand{\LCEQ}{\mathsf{LCE}}
\newcommand{\lcp}[2]{\mathsf{LCP}(#1,#2)}
\newcommand{\lcs}[2]{\mathsf{LCS}(#1,#2)}

\newcommand{\shrink}[2]{\mathit{Shrink}_{#1}^{#2}}
\newcommand{\xshrink}[2]{\mathit{XShrink}_{#1}^{#2}}
\newcommand{\pow}[2]{\mathit{Pow}_{#1}^{#2}}
\newcommand{\xpow}[2]{\mathit{XPow}_{#1}^{#2}}
\newcommand{\uniq}[1]{\mathit{Uniq}(#1)}
\newcommand{\id}[1]{\mathit{id}(#1)}
\newcommand{\encblock}[1]{\mathit{Eblock}(#1)}
\newcommand{\encblockd}[1]{\mathit{Eblock}_{d}(#1)}
\newcommand{\encpow}[1]{\mathit{Epow}(#1)}
\newcommand{\deltaLR}[1]{\Delta_{#1}}
\newcommand{\sig}[1]{\mathit{Sig}(#1)}

\newcommand{\val}[1]{\mathit{val}(#1)}
\newcommand{\valp}[1]{\mathit{val}^{+}(#1)}


%% file: abstract.tex
\begin{abstract}
In this paper, we propose a new \emph{dynamic compressed index} of $O(w)$ space for a dynamic text $T$, where 
$w = O(\min(z \log N \log^*M, N))$ is the size of the signature encoding of $T$,
$z$ is the size of the Lempel-Ziv77 (LZ77) factorization of $T$,
$N$ is the length of $T$,
and $M \geq 4N$ is an integer that can be handled in constant time under word RAM model.
Our index supports searching for a pattern $P$ in $T$ in
$O(|P| f_{\mathcal{A}} + \log w \log |P| \log^* M (\log N + \log |P| \log^* M) + \mathit{occ} \log N)$ time and
insertion/deletion of a substring of length $y$ in
$O((y+ \log N\log^* M)\log w \log N \log^* M)$ time, where
$f_{\mathcal{A}} = O(\min \{ \frac{\log\log M \log\log w}{\log\log\log M}, \sqrt{\frac{\log w}{\log\log w}} \})$.
Also, we propose a new space-efficient LZ77 factorization algorithm for
a given text of length $N$, which runs in $O(N f_{\mathcal{A}} + z \log w \log^3 N (\log^* N)^2)$ time with $O(w)$ working space.
\end{abstract}

%% file: introduction.tex
\section{Introduction}
\subsection{Dynamic compressed index}
Given a text $T$, the string indexing problem is to construct 
a data structure, called an index, so that querying
occurrences of a given pattern in $T$ can be answered efficiently.
As the size of data is growing rapidly in the last decade,
many recent studies have focused on indexes working in compressed text space
(see e.g.~\cite{GagieGKNP12,DBLP:conf/latin/GagieGKNP14,DBLP:conf/spire/ClaudeN12a,claudear:_self_index_gramm_based_compr}).
However most of them are static, i.e., they have to be reconstructed from scratch when the text is modified,
which makes difficult to apply them to a dynamic text.
Hence, in this paper, we consider the \emph{dynamic compressed text indexing problem}
of maintaining a compressed index for a text string that can be modified.
Although there exists several dynamic \emph{non-compressed} text indexes
(see e.g.~\cite{DBLP:conf/focs/SahinalpV96,DBLP:conf/soda/AlstrupBR00,DBLP:journals/jda/EhrenfeuchtMOW11} for recent work), there has been little work for the compressed variants.
Hon et al.~\cite{DBLP:conf/dcc/HonLSSY04} proposed 
the first dynamic compressed index of
$O(\frac{1}{\epsilon}(NH_0+N))$ bits of
space which supports searching of $P$ in $O(|P| \log^{2}
N(\log^{\epsilon} N + \log |\Sigma|) + \mathit{occ} \log^{1+\epsilon} N)$ time
and insertion/deletion of a substring of length $y$ in
$O((y+\sqrt{N})\log^{2+\epsilon} N)$ amortized time, where $0 <
\epsilon \leq 1$ and $H_0 \leq \log |\Sigma|$ denotes the zeroth order empirical entropy of the text of length $N$~\cite{DBLP:conf/dcc/HonLSSY04}.
Salson et al.~\cite{DBLP:journals/jda/SalsonLLM10} also proposed 
a dynamic compressed index, called \emph{dynamic FM-Index}. 
Although their approach works well in practice,
updates require $O(N \log N)$ time in the worst case.
To our knowledge, these are the only existing dynamic compressed indexes to date.

In this paper, we propose a new dynamic compressed index, as follows:
\begin{theorem}\label{theo:dynamic_index}
	Let $M$ be the maximum length of the dynamic text to index,
	$N$ the length of the current text $T$,
	$w = O(\min (z \log N \log^* M, N))$ the size of the signature encoding of $T$, and
	$z$ the number of factors in the Lempel-Ziv 77 factorization of $T$
	without self-references.
	Then, there exists a dynamic index of $O(w)$ space 
	which supports searching of a pattern $P$ 
	in 
	$O(|P| f_{\mathcal{A}} + \log w \log |P| \log^* M (\log N + \log |P| \log^* M) + \mathit{occ} \log N)$
	time, where 
	$f_{\mathcal{A}} = O(\min \{ \frac{\log\log M \log\log w}{\log\log\log M}, \sqrt{\frac{\log w}{\log\log w}} \})$,
	and insertion/deletion of a (sub)string $Y$ into/from an arbitrary position of $T$ in amortized $O((|Y|+ \log N
	\log^* M)\log w \log N \log^* M)$ time.
	Moreover, if $Y$ is given as a substring of $T$,
	we can support insertion in amortized $O(\log w (\log N \log^* M)^2)$ time.
\end{theorem}
Since $z \geq \log N$,
$\log w = \max\{\log z, \log(\log^* M)\}$.
Hence, our index is able to find pattern occurrences faster than the index of Hon et al.
when the $|P|$ term is dominating in the pattern search times.
Also, our index allows faster substring insertion/deletion on the text
when the $\sqrt{N}$ term is dominating. 

\subsubsection{Related work.}\label{sec:related_work}
To achieve the above result, technically speaking, 
we use the \emph{signature encoding} $\mathcal{G}$ of $T$, 
which is based on the \emph{locally consistent parsing} technique.
The signature encoding was proposed by Mehlhorn et al.
for equality testing on a dynamic set of strings~\cite{DBLP:journals/algorithmica/MehlhornSU97}. 
Since then, the signature encoding and the related ideas have been used in many applications. 
In particular, Alstrup et al.'s proposed dynamic index (not compressed) which is based on the signature encoding of strings,
while improving the update time of signature encodings~\cite{DBLP:conf/soda/AlstrupBR00} and the locally consistent parsing algorithm
(details can be found in the technical report~\cite{LongAlstrup}).

Our data structure uses Alstrup et al.'s
fast string concatenation/split algorithms (update algorithm) and 
linear-time computation of locally consistent parsing,
but has little else in common than those. 
Especially, Alstrup et al.'s dynamic pattern matching
algorithm~\cite{DBLP:conf/soda/AlstrupBR00,LongAlstrup}
requires to maintain specific locations called \emph{anchors}
over the parse trees of the signature encodings,
but our index does not use anchors.
Our index has close relationship to the ESP-indices~\cite{DBLP:conf/wea/TakabatakeTS14,TakabatakeTS15},
but there are two significant differences between ours and ESP-indices:
The first difference is that the ESP-index~\cite{DBLP:conf/wea/TakabatakeTS14}
is static and its online variant~\cite{TakabatakeTS15} allows only for
appending new characters to the end of the text,
while our index is fully dynamic allowing for insertion and deletion
of arbitrary substrings at arbitrary positions.
The second difference is that the pattern search time of the ESP-index 
is proportional to the number $\mathit{occ}_c$ of occurrences of the so-called ``core'' 
of a query pattern $P$, which corresponds to a maximal subtree of the 
ESP derivation tree of a query pattern $P$.
If $\mathit{occ}$ is the number of occurrences of $P$ in the text,
then it always holds that $\mathit{occ}_c \geq \mathit{occ}$,
and in general $\mathit{occ}_c$ cannot be upper bounded by
any function of $\mathit{occ}$.
In contrast, as can be seen in Theorem~\ref{theo:dynamic_index},
the pattern search time of our index is proportional to the
number $\mathit{occ}$ of occurrences of a query pattern $P$.
This became possible due to our discovery of 
a new property of the signature encoding~\cite{LongAlstrup}
(stated in Lemma~\ref{lem:pattern_occurrence_lemma1}).

As another application of signature encodings, 
Nishimoto et al. showed that signature encodings 
for a dynamic string $T$ can support Longest Common Extension (LCE) queries
on $T$ efficiently in compressed space~\cite{DBLP:journals/corr/NishimotoIIBT16} (Lemma~\ref{lem:sub_operation_lemma}). 
They also showed signature encodings can be updated in compressed space (Lemma~\ref{lem:theorem2}). 
Our algorithm uses properties of signature encodings shown in~\cite{DBLP:journals/corr/NishimotoIIBT16}, 
more precisely, Lemmas~\ref{lem:common_sequence2}-\ref{lem:sub_operation_lemma} and~\ref{lem:theorem2},  
but Lemma~\ref{lem:pattern_occurrence_lemma1} is a new property of signature encodings not described in~\cite{DBLP:journals/corr/NishimotoIIBT16}. 

In relation to our problem, 
there exists the library management problem of maintaining 
a text collection (a set of text strings) allowing for insertion/deletion of texts 
(see~\cite{DBLP:journals/corr/MunroNV15} for recent work).
While in our problem a single text is edited by insertion/deletion of substrings,
in the library management problem a text can be inserted to or deleted from the collection.
Hence, algorithms for the library management problem cannot be directly
applied to our problem. 

\subsection{Computing LZ77 factorization in compressed space.}
As an application of our dynamic compressed index,
we present a new LZ77 factorization algorithm working in compressed space.

The Lempel-Ziv77 (LZ77) factorization is defined as follows. 
\begin{definition}[Lempel-Ziv77 factorization~\cite{LZ77}]
	The Lempel-Ziv77 (LZ77) factorization of a string $s$ without self-references is 
	a sequence $f_1, \ldots, f_z$ of non-empty substrings of $s$ 
	such that $s = f_1 \cdots f_z$,
	$f_1 = s[1]$, 
	and for $1 < i \leq z$, if the character $s[|f_1..f_{i-1}|+1]$ does not occur in $s[|f_1..f_{i-1}|]$, then $f_i = s[|f_1..f_{i-1}|+1]$, otherwise $f_i$ is the longest prefix of $f_i \cdots f_z$ which occurs in $f_1 \cdots f_{i-1}$. 
	The size of the LZ77 factorization $f_1, \ldots, f_z$ of string $s$
	is the number $z$ of factors in the factorization.
\end{definition}
Although the primary use of LZ77 factorization is data compression,
it has been shown that it is a powerful tool for many string processing problems~\cite{DBLP:journals/jda/GagieGP15,DBLP:conf/latin/GagieGKNP14}.
Hence the importance of algorithms to compute LZ77 factorization is growing.
Particularly, in order to apply algorithms to large scale data,
reducing the working space is an important matter.
In this paper, we focus on LZ77 factorization algorithms working in \emph{compressed space}. 

The following is our main result.
\begin{theorem}\label{theo:lzfac} 
	Given the signature encoding $\mathcal{G}$ of size $w$ for a string $T$ of length $N$, 
	we can compute the LZ77 factorization of $T$ in 
	$O(z \log w \log ^3 N(\log^* M)^2)$ time and $O(w)$ working space 
	where $z$ is the size of the LZ77 factorization of $T$.
\end{theorem}
In~\cite{DBLP:journals/corr/NishimotoIIBT16}, it was shown that 
the signature encoding $\mathcal{G}$ can be constructed efficiently from various types of inputs,
in particular, in $O(N f_{\mathcal{A}})$ time and $O(w)$ working space from uncompressed string $T$.
Therefore we can compute LZ77 factorization of a given $T$ of length $N$ 
in $O(N f_\mathcal{A} + z \log w \log ^3 N (\log^* M)^2)$ time and $O(w)$ working space.

\subsubsection{Related work.}
Goto et al.~\cite{DBLP:journals/corr/abs-1107-2729} showed how,
given the grammar-like representation for string $T$ generated
by the LCA algorithm~\cite{DBLP:journals/ieicet/SakamotoMKS09},
to compute the LZ77 factorization of $T$
in $O(z \log^2 m \log^3 N + m \log m \log^3 N)$ time and $O(m \log^2 m)$ space, 
where $m$ is the size of the given representation.
Sakamoto et al.~\cite{DBLP:journals/ieicet/SakamotoMKS09} claimed that
$m = O(z \log N \log^* N)$, however, it seems that in this bound they do not consider 
the production rules to represent maximal runs of non-terminals in the derivation tree.
The bound we were able to obtain with the best of our knowledge and 
understanding is $m = O(z \log^2 N \log^* N)$,
and hence our algorithm seems to use less space than the algorithm of Goto et al.~\cite{DBLP:journals/corr/abs-1107-2729}.
Recently, Fischer et al.~\cite{FGGK15} showed
a Monte-Carlo randomized algorithms to compute  
an approximation of the LZ77 factorization
with at most $2z$ factors in $O(N \log N)$ time,
and another approximation with at most $(i+\epsilon)z$ factors
in $O(N \log^2 N)$ time for any constant $\epsilon > 0$,
using $O(z)$ space each.

Another line of research is LZ77 factorization working in compressed space in terms of Burrows-Wheeler transform (BWT) based methods.
Policriti and Prezza recently proposed algorithms running in
$N H_0 + o(N \log |\Sigma|) + O(|\Sigma| \log N)$ bits of space and $O(N \log N)$ time~\cite{DBLP:conf/spire/PolicritiP15},
or $O(R \log N)$ bits of space and $O(N \log R)$ time~\cite{DBLP:conf/dcc/PolicritiP16},
where $R$ is the number of runs in the BWT of the reversed string of $T$.
Because their and our algorithms are established on different measures of compression,
they cannot be easily compared.
For example, our algorithm is more space efficient than the algorithm in~\cite{DBLP:conf/dcc/PolicritiP16} when $w = o(R)$,
but it is not clear when it happens.

%% file: preliminaries.tex
\section{Preliminaries} \label{sec:preliminary}

\subsection{Strings}

Let $\Sigma$ be an ordered alphabet.
An element of $\Sigma^*$ is called a string.
For string $w = xyz$,
$x$, $y$ and $z$ are called a prefix, substring, and suffix of $w$, respectively.
The length of string $w$ is denoted by $|w|$.
The empty string $\varepsilon$ is a string of length $0$.
Let $\Sigma^+ = \Sigma^* - \{\varepsilon\}$.
For any $1 \leq i \leq |w|$, $w[i]$ denotes the $i$-th character of $w$.
For any $1 \leq i \leq j \leq |w|$,
$w[i..j]$ denotes the substring of $w$ that begins at position $i$
and ends at position $j$.
Let $w[i..] = w[i..|w|]$ and $w[..i] = w[1..i]$ for any $1 \leq i \leq |w|$.
For any string $w$, let $\rev{w}$ denote the reversed string of $w$,
that is, $\rev{w} = w[|w|] \cdots w[2]w[1]$. 
For any strings $w$ and $u$, 
let $\lcp{w}{u}$ (resp. $\lcs{w}{u}$) denote the length of 
the longest common prefix (resp. suffix) of $w$ and $u$.
Given two strings $s_1, s_2$ and two integers $i, j$, let $\LCEQ(s_1, s_2, i, j)$ denote a query which returns $\lcp{s_1[i..|s_1|]}{s_2[j..|s_2|]}$.
For any strings $p$ and $s$, let $\Occ(p,s)$
denote all occurrence positions of $p$ in $s$,
namely, $\Occ(p, s) = \{i \mid p = s[i..i+|p|-1], 1 \leq i \leq |s|-|p|+1\}$.
Our model of computation is the unit-cost word RAM with machine word 
size of $\Omega(\log_2 M)$ bits,
and space complexities will be evaluated by the number of machine words.
Bit-oriented evaluation of space complexities can be obtained with 
a $\log_2 M$ multiplicative factor.

\subsection{Context free grammars as compressed representation of strings}

\vspace*{0.5pc}
\noindent \textbf{Straight-line programs.}
A \emph{straight-line program} (\emph{SLP}) is a context free grammar in the Chomsky normal form
that generates a single string.
Formally, an SLP that generates $T$ is
a quadruple $\mathcal{G} = (\Sigma, \mathcal{V}, \mathcal{D}, S)$, 
such that
$\Sigma$ is an ordered alphabet of terminal characters;
$\mathcal{V} = \{ X_1, \ldots, X_{n} \}$ is a set of positive integers, called \emph{variables};
$\mathcal{D} = \{X_i \rightarrow \mathit{expr}_i\}_{i = 1}^{n}$ is a set of \emph{deterministic productions} (or \emph{assignments})
with each $\mathit{expr}_i$ being either of form $X_\ell X_r~(1 \leq \ell, r < i)$, or a single character $a \in \Sigma$;
and $S := X_{n} \in \mathcal{V}$ is the start symbol which derives the string $T$.
We also assume that the grammar neither contains \emph{redundant} variables (i.e., there is at most one assignment whose righthand side is $\mathit{expr}$)
nor \emph{useless} variables (i.e., every variable appears at least once in the derivation tree of $\mathcal{G}$).
The \emph{size} of the SLP $\mathcal{G}$ is the number $n$ of productions in $\mathcal{D}$.
In the extreme cases the length $N$ of the string $T$ can be as large as $2^{n-1}$,
however, it is always the case that $n \geq \log_2 N$.
See also Example~\ref{ex:SLP}.

Let $\mathit{val}: \mathcal{V} \rightarrow \Sigma^+$ be the function
which returns the string derived by an input variable.
If $s = \val{X}$ for $X \in \mathcal{V}$,
then we say that the variable $X$ \emph{represents} string $s$.
For any variable sequence $y \in \mathcal{V}^{+}$,
let $\valp{y} = \val{y[1]} \cdots \val{y[|y|]}$.
For any variable $X_i$ with $X_i \rightarrow X_{\ell}X_{r} \in \mathcal{D}$, 
let $X_i.{\rm left} = \val{X_{\ell}}$ and $X_i.{\rm right} = \val{X_{r}}$,
which are called 
the \emph{left string} and the \emph{right string} of $X_i$, respectively.
For two variables $X_i, X_j \in \mathcal{V}$, we say that $X_i$ occurs at position $c$ in $X_j$ 
if there is a node labeled with $X_i$ in the derivation tree of $X_j$
and the leftmost leaf of the subtree rooted at that node labeled with $X_i$
is the $c$-th leaf in the derivation tree of $X_j$. 
We define the function $\vOcc(X_i,X_j)$ which returns all positions of $X_i$ in the derivation tree of $X_j$. 

\vspace*{0.5pc}
\noindent \textbf{Run-length straight-line programs.}
We define \emph{run-length SLPs}, (\emph{RLSLPs}) as an extension to SLPs,
which allow \emph{run-length encodings} in the righthand sides of productions, i.e.,
$\mathcal{D}$ might contain a production $X \rightarrow \hat{X}^{k} \in \mathcal{V} \times \mathcal{N}$.
The \emph{size} of the RLSLP is still the number of productions in $\mathcal{D}$
as each production can be encoded in constant space.
Let $\mathit{Assgn}_{\mathcal{G}}$ be the function such that
$\mathit{Assgn}_{\mathcal{G}}(X_i) = \mathit{expr_i}$ iff $X_i \rightarrow \mathit{expr_i} \in \mathcal{D}$.
Also, let $\mathit{Assgn}^{-1}_{\mathcal{G}}$ denote the reverse function of $\mathit{Assgn}_{\mathcal{G}}$.
When clear from the context, 
we write $\mathit{Assgn}_{\mathcal{G}}$ and $\mathit{Assgn}^{-1}_{\mathcal{G}}$
as $\mathit{Assgn}$ and $\mathit{Assgn}^{-1}$, respectively.

	We define the left and right strings for any variable
	$X_i \rightarrow X_{\ell}X_r \in \mathcal{D}$ in a similar way to SLPs. 
	Furthermore, for any $X \rightarrow \hat{X}^{k} \in \mathcal{D}$,
	let $X.{\rm left} = \val{\hat{X}}$ and $X.{\rm right} = \val{\hat{X}}^{k-1}$.

\vspace*{0.5pc}
\noindent \textbf{Representation of RLSLPs.}
For an RLSLP $\mathcal{G}$ of size $w$, 
we can consider a DAG of size $w$ as a compact representation of the derivation trees of variables in $\mathcal{G}$.
Each node represents a variable $X$ in $\mathcal{V}$ and stores $|\val{X}|$
and out-going edges represent the assignments in $\mathcal{D}$:
For an assignment $X_i \rightarrow X_{\ell}X_r \in \mathcal{D}$,
there exist two out-going edges from $X_i$ to its ordered children $X_{\ell}$ and $X_r$;
and for $X \rightarrow \hat{X}^{k} \in \mathcal{D}$, 
there is a single edge from $X$ to $\hat{X}$ with the multiplicative factor $k$. 
	For $X \in \mathcal{V}$, let $\mathit{parents}(X)$ be the set of variables 
	which have out-going edge to $X$ in the DAG of $\mathcal{G}$.
	To compute $\mathit{parents}(X)$ for $X \in \mathcal{V}$ in linear time, 
	we let $X$ have a doubly-linked list of length $|\mathit{parents}(X)|$ to represent $\mathit{parents}(X)$:
	Each element is a pointer to a node for $X' \in \mathit{parents}(X)$ (the order of elements is arbitrary).
	Conversely, we let every parent $X'$ of $X$ have the pointer to the corresponding element in the list.
See also Example~\ref{ex:tree}.

%% file: framework.tex
\section{Signature encoding}\label{sec:Framework}

Here, we recall the \emph{signature encoding} 
first proposed by Mehlhorn et al.~\cite{DBLP:journals/algorithmica/MehlhornSU97}.
Its core technique is \emph{locally consistent parsing} defined as follows:

\begin{lemma}[Locally consistent parsing~\cite{DBLP:journals/algorithmica/MehlhornSU97,LongAlstrup}]\label{lem:CoinTossing}
	Let $W$ be a positive integer.
	There exists a function $f: [0..W]^{\log^* W + 11} \rightarrow \{0,1\}$ such that,
	for any $p \in [1..W]^n$ with $n \geq 2$ and $p[i] \neq p[i+1]$ for any $1 \leq i < n$,
	the bit sequence $d$ defined by 
	$d[i] = f(\tilde{p}[i-\deltaLR{L}], \ldots, \tilde{p}[i+\deltaLR{R}])$ for $1 \leq i \leq n$, satisfies:
	$d[1] = 1$; 
	$d[n] = 0$;
	$d[i] + d[i+1] \leq 1$ for $1 \leq i < n$; 
	and $d[i] + d[i+1] + d[i+2] + d[i+3] \geq 1$ for any $1 \leq i < n-3$;
	where $\deltaLR{L} = \log ^*W + 6$, $\deltaLR{R} = 4$, and $\tilde{p}[j] = p[j]$ for all $1 \leq j \leq n$, $\tilde{p}[j] = 0$ otherwise. 
	Furthermore, we can compute $d$ in $O(n)$ time using a precomputed table of size $o(\log W)$, which can be computed in $o(\log W)$ time.
\end{lemma}

For the bit sequence $d$ of Lemma~\ref{lem:CoinTossing},
we define the function $\mathit{Eblock}_d(p)$ that 
decomposes an integer sequence $p$ according to $d$:
$\encblockd{p}$ decomposes $p$ into a sequence
$q_1, \ldots, q_j$ of substrings called \emph{blocks} of $p$,
such that $p = q_1 \cdots q_j$ and 
$q_i$ is in the decomposition iff $d[|q_1 \cdots q_{i-1}|+1] = 1$
for any $1 \leq i \leq j$.
Note that each block is of length from two to four by the property of $d$, i.e.,
$2 \leq |q_i| \leq 4$ for any $1 \leq i \leq j$.
Let $|\encblockd{p}| = j$ and let $\encblockd{s}[i] = q_i$.
We omit $d$ and write $\encblock{p}$ when it is clear from the context, 
and we use implicitly the bit sequence created by Lemma~\ref{lem:CoinTossing} as $d$. 

We complementarily use run-length encoding to get a sequence to which $\mathit{Eblock}$ can be applied.
Formally, for a string $s$, 
let $\encpow{s}$ be the function which 
groups each maximal run of same characters $a$ as $a^k$,
where $k$ is the length of the run. 
$\encpow{s}$ can be computed in $O(|s|)$ time.
Let $|\encpow{s}|$ denote the number of maximal runs of same characters in $s$ and let $\encpow{s}[i]$ denote $i$-th maximal run in $s$.
See also Example~\ref{ex:Encblock}.

The signature encoding is the RLSLP $\mathcal{G} = (\Sigma, \mathcal{V}, \mathcal{D}, S)$,
where the assignments in $\mathcal{D}$ are determined by
recursively applying $\mathit{Eblock}$ and $\mathit{Epow}$ to $T$
until a single integer $S$ is obtained.
We call each variable of the signature encoding a \emph{signature},
and use $e$ (for example, $e_i \rightarrow e_{\ell}e_{r} \in \mathcal{D}$)
instead of $X$ to distinguish from general RLSLPs.

For a formal description,
let $E := \Sigma \cup \mathcal{V}^{2} \cup \mathcal{V}^{3} \cup \mathcal{V}^{4} \cup (\mathcal{V} \times \mathcal{N})$
and let $\mathit{Sig}: E \rightarrow \mathcal{V}$ be the function such that:
$\mathit{Sig}(\mathit{x}) = e$ if $(e \rightarrow \mathit{x}) \in \mathcal{D}$;
$\mathit{Sig}(\mathit{x}) = \mathit{Sig}( \mathit{Sig}(\mathit{x}[1..|\mathit{x}|-1]) \mathit{x}[|\mathit{x}|])$
if $\mathit{x} \in \mathcal{V}^{3} \cup \mathcal{V}^{4}$;
or otherwise undefined.
Namely, the function $\mathit{Sig}$ returns, if any,
the lefthand side of the corresponding production of $\mathit{x}$
by recursively applying the $\mathit{Assgn}^{-1}$ function from left to right.
For any $p \in E^*$,
let $\mathit{Sig}^{+}(p) = \sig{p[1]} \cdots \sig{p[|p|]}$.

The signature encoding of string $T$ is defined by the following $\mathit{Shrink}$ and $\mathit{Pow}$ functions:
$\shrink{t}{T} = \mathit{Sig}^{+}(T)$ for $t = 0$, and $\shrink{t}{T} = \mathit{Sig}^{+}(\encblock{\pow{t-1}{T}})$ for $0 < t \leq h$; and
$\pow{t}{T} = \mathit{Sig}^{+}(\encpow{\shrink{t}{T}})$ for $0 \leq t \leq h$;
where $h$ is the minimum integer satisfying $|\pow{h}{T}| = 1$.
Then, the start symbol of the signature encoding is $S = \pow{h}{T}$.
We say that a node is in \emph{level} $t$ in the derivation tree of $S$
if the node is produced by $\shrink{t}{T}$ or $\pow{t}{T}$.
The height of the derivation tree of the signature encoding of $T$ is
$O(h) = O(\log |T|)$.
For any $T \in \Sigma^+$,
let $\id{T} = \pow{h}{T} = S$, i.e.,
the integer $S$ is the signature of $T$.
We let $N \leq M/4$. 
More specifically, $M = 4N$ if $T$ is static, 
and $M/4$ is the upper bound of the length of $T$ if we consider updating $T$ dynamically. 
Since all signatures are in $[1..M-1]$, we set $W = M$ in Lemma~\ref{lem:CoinTossing} used by the signature encoding.
In this paper, we implement signature encodings by the DAG of RLSLP introduced in Section~\ref{sec:preliminary}.
See also Example~\ref{ex:signature_dictionary} and Figure~\ref{fig:SignatureTree}.

\subsection{Commmon sequences}\label{sec:CommonSeq}

Here, we recall the most important property of the signature encoding,
which ensures the existence of common signatures to all occurrences of same substrings by the following lemma.

\begin{lemma}[common sequences~\cite{18045,DBLP:journals/corr/NishimotoIIBT16}]\label{lem:common_sequence2}
	Let $\mathcal{G} = (\Sigma, \mathcal{V}, \mathcal{D}, S)$ be a signature encoding for a string $T$. 
	Every substring $P$ in $T$ is represented 
	by a signature sequence $\mathit{Uniq}(P)$ in $\mathcal{G}$ for a string $P$, 
	where $|\encpow{\uniq{P}}| = O(\log |P| \log^* M)$.  
\end{lemma}
$\mathit{Uniq}(P)$, which we call the \emph{common sequence} of $P$, is defined by the following.

\begin{definition}\label{def:xshrink}
For a string $P$, let
\begin{eqnarray*}
 \mathit{XShrink}_t^{P} &=&
  \begin{cases}
   \mathit{Sig}^{+}(P) & \mbox{ for } t = 0, \\
   \mathit{Sig}^{+}(\encblockd{\mathit{XPow}_{t-1}^{P}}[|L_{t}^{P}|..|\mathit{XPow}_{t-1}^{P}|-|R_{t}^{P}|]) & \mbox{ for } 0 < t \leq h^{P}, \\
  \end{cases} \\
\mathit{XPow}_t^{P} &=& \mathit{Sig}^{+}(\encpow{\mathit{XShrink}_t^{P}[|\hat{L}_{t}^{P}| + 1..|\mathit{XShrink}_t^{P}| - |\hat{R}_{t}^{P}]}|) \ \mbox{ for } 0 \leq t < h^{P}, \mbox{ where}
\end{eqnarray*}
\begin{itemize}
  \item $L_{t}^{P}$ is the shortest prefix of $\mathit{XPow}_{t-1}^{P}$ of length at least $\deltaLR{L}$ such that $d[|L_{t}^{P}|+1] = 1$,
  \item $R_{t}^{P}$ is the shortest suffix of $\mathit{XPow}_{t-1}^{P}$ of length at least $\deltaLR{R}+1$ such that $d[|d| - |R_{t}^{P}| + 1] = 1$,
  \item $\hat{L}_{t}^{P}$ is the longest prefix of $\mathit{XShrink}_t^{P}$ such that $|\encpow{\hat{L}_{t}^{P}}|  = 1$,
  \item $\hat{R}_{t}^{P}$ is the longest suffix of $\mathit{XShrink}_t^{P}$ such that $|\encpow{\hat{R}_{t}^{P}}| = 1$, and
  \item $h^{P}$ is the minimum integer such that $|\encpow{\mathit{XShrink}_{h^{P}}^{P}}| \leq \Delta_{L} + \Delta_{R} + 9$.
\end{itemize}
\end{definition}
Note that $\Delta_{L} \leq |L_{t}^{P}| \leq \Delta_{L} + 3$ and $\Delta_{R}+1 \leq |R_{t}^{P}| \leq \Delta_{R} + 4$ hold by the definition. 
Hence $|\xshrink{t+1}{P}| > 0$ holds if $|\encpow{\xshrink{t}{P}}| > \Delta_{L} + \Delta_{R} + 9$. 
Then,
\[ 
\mathit{Uniq}(P) = \hat{L}_{0}^{P}L_{0}^{P} \cdots 
\hat{L}_{h^{P}-1}^{P}L_{h^{P}-1}^{P}\mathit{XShrink}_{h^{P}}^{P}R_{h^{P}-1}^{P}\hat{R}_{h^{P}-1}^{P} \cdots R_{0}^{P}\hat{R}_{0}^{P}.
\] 

	We give an intuitive description of Lemma~\ref{lem:common_sequence2}. 
	Recall that the locally consistent parsing of Lemma~\ref{lem:CoinTossing}. 
	Each $i$-th bit of bit sequence $d$ of Lemma~\ref{lem:CoinTossing} for a given string $s$
	is determined by $s[i-\deltaLR{L}..i+\deltaLR{R}]$. 
	Hence, for two positions $i, j$ such that $P = s[i..i+k-1] = s[j..j+k-1]$ for some $k$, 
	$d[i+\deltaLR{L}..i+k-1-\deltaLR{R}] = d[j+\deltaLR{L}..j+k-1-\deltaLR{R}]$ holds, namely, 
	``internal'' bit sequences of the same substring of $s$ are equal. 
	Since each level of the signature encoding uses the bit sequence,
	all occurrences of same substrings in a string share same internal signature sequences,
	and this goes up level by level.
	$\xshrink{t}{P}$ and $\xpow{t}{P}$ represent signature sequences 
	which are obtained from only internal signature sequences of $\xpow{t-1}{T}$ and $\xshrink{t}{T}$, respectively. 
	This means that $\xshrink{t}{P}$ and $\xpow{t}{P}$ are always created over $P$.
	From such common signatures we take as short signature sequence as possible for $\mathit{Uniq}(P)$:
	Since $\valp{\pow{t-1}{P}} = \valp{L_{t-1}^{P}\xshrink{t}{P}R_{t-1}^{P}}$ and 
	$\valp{\shrink{t}{P}} = \valp{\hat{L}_{t}^{P}\xpow{t}{P}\hat{R}_{t}^{P}}$ hold, 
	$|\encpow{\mathit{Uniq}(P)}| = O(\log |P| \log^* M)$ and $\valp{\mathit{Uniq}(P)} = P$ hold.
	Hence Lemma~\ref{lem:common_sequence2} holds~(see also Figure~\ref{fig:CommonSequence})~\footnote{
		The common sequences are conceptually equivalent to
		the \emph{cores}~\cite{maruyama13:_esp} which are defined for the
		\emph{edit sensitive parsing} of a text,
		a kind of locally consistent parsing of the text.
	}. 

From the common sequences we can derive many useful properties of signature encodings like listed below
(see the references for proofs).

The number of ancestors of nodes corresponding to $\mathit{Uniq}(P)$ is upper bounded by:
\begin{lemma}[\cite{DBLP:journals/corr/NishimotoIIBT16}]\label{lem:ancestors}
  Let $\mathcal{G}$ be a signature encoding for a string $T$, 
  $P$ be a string, and let $\mathcal{T}$ be the derivation tree of a signature $e \in \mathcal{V}$. 
  Consider an occurrence of $P$ in $s$,
  and the induced subtree $X$ of $\mathcal{T}$
  whose root is the root of $\mathcal{T}$ and whose leaves are 
  the parents of the nodes representing $\uniq{P}$, where $s = \val{e}$.
  Then $X$ contains $O(\log^* M)$ nodes for every level and
  $O(\log |s| + \log |P| \log^* M)$ nodes in total.
\end{lemma}

We can efficiently compute $\uniq{P}$ for a substring $P$ of $T$.
\begin{lemma}[\cite{DBLP:journals/corr/NishimotoIIBT16}]\label{lem:ComputeShortCommonSequence}
Using a signature encoding $\mathcal{G}$ of size $w$, 
given a signature $e \in \mathcal{V}$ (and its corresponding node in the DAG)
and two integers $j$ and $y$, 
we can compute $\encpow{\uniq{s[j..j+y-1]}}$ in $O(\log |s| + \log y \log^* M)$ time, 
where $s = \val{e}$.
\end{lemma}

The next lemma shows that $\mathcal{G}$ requires only \emph{compressed space}:
\begin{lemma}[\cite{18045,DBLP:journals/corr/NishimotoIIBT16}]\label{lem:upperbound_signature}
The size $w$ of the signature encoding of $T$ of length $N$ is 
$O(\min(z \log N \log^* M, N))$,
where $z$ is the number of factors in the LZ77 factorization without self-reference of $T$. 
\end{lemma}

The next lemma shows that the signature encoding supports 
(both forward and backward) LCE queries on a given arbitrary pair of signatures.
\begin{lemma}[\cite{DBLP:journals/corr/NishimotoIIBT16}]\label{lem:sub_operation_lemma}
	Using a signature encoding $\mathcal{G}$ for a string $T$, 
	we can support queries $\LCEQ(s_1, s_2, i, j)$ and $\LCEQ(s_1^{R}, s_2^{R}, i, j)$ 
	in $O(\log |s_1| + \log |s_2| + \log\ell \log^* M)$ time 
	for given two signatures $e_1, e_2 \in \mathcal{V}$ and 
	two integers $1 \leq i \leq |s_1|$, $1 \leq j \leq |s_2|$, 
	where $s_1 = \val{e_1}$, $s_2 = \val{e_2}$ and $\ell$ is the answer to the $\LCEQ$ query. 
\end{lemma}

%% file: update.tex
\subsection{Dynamic signature encoding}

We consider a \emph{dynamic signature encoding} $\mathcal{G}$ of $T$,
which allows for efficient updates of $\mathcal{G}$ in compressed space according to the following operations:
$\mathit{INSERT}(Y, i)$ inserts a string $Y$ into $T$ at position $i$, i.e., $T \leftarrow T[..i-1] Y T[i..]$;
$\mathit{INSERT'}(j, y, i)$ inserts $T[j..j+y-1]$ into $T$ at position $i$, i.e., $T \leftarrow T[..i-1] T[j..j+y-1] T[i..]$; and
$\mathit{DELETE}(j, y)$ deletes a substring of length $y$ starting at $j$, i.e., $T \leftarrow T[..j-1]T[j+y..]$.

During updates we recompute $\shrink{t}{T}$ and $\pow{t}{T}$ for some part of new $T$
(note that the most part is unchanged thanks to the virtue of signature encodings, Lemma~\ref{lem:ancestors}).
When we need a signature for $\mathit{expr}$, 
we look up the signature assigned to $\mathit{expr}$ (i.e., compute $\mathit{Assign}^{-1}(\mathit{expr})$) and use it if such exists.
If $\mathit{Assign}^{-1}(\mathit{expr})$ is undefined
we create a new signature $e_{\mathit{new}}$, which is an integer that is currently not used as signatures,
and add $e_{\mathit{new}} \rightarrow \mathit{expr}$ to $\mathcal{D}$.
Also, updates may produce a useless signature whose parents in the DAG are all removed.
We remove such useless signatures from $\mathcal{G}$ during updates.

We can upper bound the number of signatures added to or removed from $\mathcal{G}$ after a single update operation by the following lemma.
\footnote{The property is used in~\cite{DBLP:journals/corr/NishimotoIIBT16}, but there is no corresponding lemma to state it clearly.}
\begin{lemma}\label{lem:up_sig_bound}
After $\mathit{INSERT}(Y, i)$ or $\mathit{DELETE}(j, y)$ operation,
$O(y + \log N \log^* M)$ signatures are added to or removed from $\mathcal{G}$,
where $|Y| = y$.
After $\mathit{INSERT'}(j, y, i)$ operation,
$O(\log N \log^* M)$ signatures are added to or removed from $\mathcal{G}$.
\end{lemma}
\begin{proof}
Consider $\mathit{INSERT'}(j, y, i)$ operation.
Let $T' = T[..i-1] T[j..j+y-1] T[i..]$ be the new text.
Note that by Lemma~\ref{lem:common_sequence2} the signature encoding of $T'$ is created over $\uniq{T[..i-1]} \uniq{T[j..j+y-1]} \uniq{T[i..]}$,
and hence, $O(\log N \log^* M)$ signatures can be added by Lemma~\ref{lem:ancestors}.
Also, $O(\log N \log^* M)$ signatures, which were created over $\uniq{T[..i-1]} \uniq{T[i..]}$, may be removed.

For $\mathit{INSERT}(Y, i)$ operation,
we additionally think about the possibility that $O(y)$ signatures are added to create $\uniq{Y}$.
Similarly, for $\mathit{DELETE}(j, y)$ operation, $O(y)$ signatures, which are used in and under $\uniq{T[j..j+y-1]}$, can be removed.
\qed
\end{proof}

In~\cite{DBLP:journals/corr/NishimotoIIBT16}, it was shown how to augment the DAG representation of $\mathcal{G}$ 
to add/remove an assignment to/from $\mathcal{G}$ in $O(f_{\mathcal{A}})$ time, where 
$f_{\mathcal{A}} = O\left(\min \left\{ \frac{\log \log M \log \log w}{\log \log \log M}, \sqrt{\frac{\log w}{\log\log w}} \right\} \right)$
is the time complexity of Beame and Fich's data structure~\cite{DBLP:journals/jcss/BeameF02}
to support predecessor/successor queries on a set of $w$ integers from an $M$-element universe.\footnote{
	The data structure is, for example, used to compute $\mathit{Assgn}^{-1}(\cdot)$.
	Alstrup et al.~\cite{LongAlstrup} used hashing for this purpose.
	However,
	since we are interested in the worst case time complexities, we use
	the data structure~\cite{DBLP:journals/jcss/BeameF02} in place of hashing.
}
Note that there is a small difference in our DAG representation from the one in~\cite{DBLP:journals/corr/NishimotoIIBT16};
our DAG has a doubly-linked list representing the parents of a node.
We can check if a signature is useless or not by checking if the list is empty or not,
and the lists can be maintained in constant time after adding/removing an assignment.
Hence, the next lemma still holds for our DAG representation.
\begin{lemma}[Dynamic signature encoding~\cite{DBLP:journals/corr/NishimotoIIBT16}]\label{lem:theorem2}
	After processing $\mathcal{G}$ in $O(w f_{\mathcal{A}})$ time,
	we can insert/delete any (sub)string $Y$ of length $y$ into/from an arbitrary position of $T$ in $O((y+ \log N\log^* M) f_{\mathcal{A}})$ time.
	Moreover, if $Y$ is given as a substring of $T$,
	we can support insertion in $O(f_{\mathcal{A}} \log N \log^* M)$ time.
\end{lemma}

%% file: static_index.tex
\section{Dynamic Compressed Index}\label{sec:static_section}

In this section, we present our dynamic compressed index based on signature encoding.
As already mentioned in the introduction, 
our strategy for pattern matching is different from that of Alstrup et al.~\cite{LongAlstrup}.
It is rather similar to the one taken in the static index for SLPs of Claude and Navarro~\cite{claudear:_self_index_gramm_based_compr}.
Besides applying their idea to RLSLPs, we show how to speed up pattern matching by utilizing the properties of signature encodings.
\vspace*{0.5pc}
\noindent \textbf{Index for SLPs.}
Here we review how the index in~\cite{claudear:_self_index_gramm_based_compr}
for SLP $\mathcal{S}$ generating a string $T$ computes $\Occ(P,T)$ for a given string $P$.
The key observation is that, any occurrence of $P$ in $T$
can be uniquely associated with the lowest node that covers the occurrence of $P$ in the derivation tree.
As the derivation tree is binary, if $|P| > 1$, then the node is labeled with some variable $X \in \mathcal{V}$ such that
$P_1$ is a suffix of $X.{\rm left}$ and $P_2$ is a prefix of $X.{\rm right}$, where $P = P_1P_2$ with $1 \leq |P_1| < |P|$.
Here we call the pair $(X, |X.{\rm left}| - |P_1| + 1)$ a \emph{primary occurrence} of $P$, and 
let $\mathit{pOcc}_{\mathcal{S}}(P, j)$ denote the set of such primary occurrences with $|P_1| = j$.
The set of all primary occurrences is denoted by $\mathit{pOcc}_{\mathcal{S}}(P) = \bigcup_{1 \leq j < |P|} \mathit{pOcc}_{\mathcal{S}}(P, j)$.
Then, we can compute $\Occ(P, T)$ by first computing primary occurrences and
enumerating the occurrences of $X$ in the derivation tree.

The set $\mathit{Occ}(P,T)$ of occurrences of $P$ in $T$ is represented by $\mathit{pOcc}_{\mathcal{S}}(P)$ as follows:
$\mathit{Occ}(P,T) = \{ j+k-1 \mid (X, j) \in \mathit{pOcc}_{\mathcal{S}}(P), k \in \mathit{vOcc}(X, S)\}$ if $|P| > 1$;
$\mathit{Occ}(P,T) = \mathit{vOcc}(X ,S) ((X \rightarrow P) \in \mathcal{D})$ if $|P| = 1$. 
See also Example~\ref{ex:primary_occurrence}. 

Hence the task is to compute $\mathit{pOcc}_{\mathcal{S}}(P)$ and $\mathit{vOcc}(X,S)$ efficiently.
Note that $\mathit{vOcc}(X,S)$ can be computed in $O(|\mathit{vOcc}(X,S)| h)$ time
by traversing the DAG in a reversed direction from $X$ to the source,
where $h$ is the height of the derivation tree of $S$.
Hence, in what follows, we explain how to compute $\mathit{pOcc}_{\mathcal{S}}(P)$ for a string $P$ with $|P| > 1$.
We consider the following problem:
\begin{problem}[Two-Dimensional Orthogonal Range Reporting Problem]\label{problem:2D}
Let $\mathcal{X}$ and $\mathcal{Y}$ denote subsets of two ordered sets,
and let $\mathcal{R} \subseteq \mathcal{X} \times \mathcal{Y}$ be a set of points on the two-dimensional plane,
where $|\mathcal{X}|, |\mathcal{Y}| \in O(|\mathcal{R}|)$.
A data structure for this problem supports a query $\mathit{report}_{\mathcal{R}}(x_1,x_2,y_1,y_2)$;
given a rectangle $(x_1,x_2,y_1,y_2)$ with $x_1, x_2 \in \mathcal{X}$ and $y_1, y_2 \in \mathcal{Y}$, 
returns $\{ (x,y) \in \mathcal{R} \mid x_1 \leq x \leq x_2, y_1 \leq y \leq y_2 \}$.
\end{problem}

Data structures for Problem~\ref{problem:2D} are widely studied in computational geometry.
There is even a dynamic variant, which we finally use for our dynamic index.
Until then, we just use any data structure that occupies $O(|\mathcal{R}|)$ space and
supports queries in $O(\hat{q}_{|\mathcal{R}|} + q_{|\mathcal{R}|} \mathit{qocc})$ time with $\hat{q}_{|\mathcal{R}|} = O(\log |\mathcal{R}|)$,
where $\mathit{qocc}$ is the number of points to report.

Now, given an SLP $\mathcal{S}$, we consider a two-dimensional plane defined by 
$\mathcal{X} = \{ X.{\rm left}^{R} \mid X \in \mathcal{V} \}$ and $\mathcal{Y}  = \{ X.{\rm right} \mid X \in \mathcal{V} \}$,
where elements in $\mathcal{X}$ and $\mathcal{Y}$ are sorted by lexicographic order.
Then consider a set of points $\mathcal{R} = \{ (X.{\rm left}^{R}, X.{\rm right}) \mid X \in \mathcal{V} \}$.
For a string $P$ and an integer $1 \leq j < |P|$,
let $y_1^{(P,j)}$ (resp. $y_2^{(P,j)}$) denote the lexicographically smallest (resp. largest) element 
in $\mathcal{Y}$ that has $P[j+1..]$ as a prefix.
If there is no such element, it just returns NIL and we can immediately know that $\mathit{pOcc}_{\mathcal{S}}(P,j) = \emptyset$.
We define $x_1^{(P,j)}$ and $x_2^{(P,j)}$ in a similar way over $\mathcal{X}$.
Then, $\mathit{pOcc}_{\mathcal{S}}(P, j)$ can be computed by a query $\mathit{report}_{\mathcal{R}}(x_1^{(P,j)}, x_2^{(P,j)}, y_1^{(P,j)}, y_2^{(P,j)})$ (see also Example~\ref{ex:SLPGrid}).

Using this idea, we can get the next result:
\begin{lemma}\label{lem:LinearSpaceIndexOfSLP}
For an SLP $\mathcal{S}$ of size $n$,
there exists a data structure of size $O(n)$ that computes, given a string $P$,
$\mathit{pOcc}_{\mathcal{S}}(P)$ in $O(|P| (h + |P|) \log n + q_{n} |\mathit{pOcc}_{\mathcal{S}}(P)|)$ time.
\end{lemma}
\begin{proof}
	For every $1 \leq j < |P|$, we compute $\mathit{pOcc}_{\mathcal{S}}(P, j)$ by 
	$\mathit{report}_{\mathcal{R}}(x_1^{(P,j)}, x_2^{(P,j)}, y_1^{(P,j)}, y_2^{(P,j)})$.
	We can compute $y_1^{(P,j)}$ and $y_2^{(P,j)}$ in $O((h + |P|) \log n)$ time by binary search on $\mathcal{Y}$,
	where each comparison takes $O(h + |P|)$ time for expanding the first $O(|P|)$ characters of variables subjected to comparison.
	In a similar way, $x_1^{(P,j)}$ and $x_2^{(P,j)}$ can be computed in $O((h + |P|) \log n)$ time.
	Thus, the total time complexity is 
	$O(|P| ((h + |P|) \log n + \hat{q}_{n}) + q_{n} |\mathit{pOcc}_{\mathcal{S}}(P)|) = O(|P| (h + |P|) \log n + q_{n} |\mathit{pOcc}_{\mathcal{S}}(P)|)$.
	\qed
\end{proof}

\vspace*{0.5pc}
\noindent \textbf{Index for RLSLPs.}
We extend the idea for the SLP index described above to RLSLPs.
The difference from SLPs is that we have to deal with occurrences of $P$
that are covered by a node labeled with $X \rightarrow \hat{X}^{k}$ but not covered by any single child of the node in the derivation tree.
In such a case, there must exist $P = P_1P_2$ with $1 \leq |P_1| < |P|$ such that
$P_1$ is a suffix of $X.{\rm left} = \valp{\hat{X}}$ and $P_2$ is a prefix of $X.{\rm right} = \valp{\hat{X}^{k-1}}$.
Let $j = |\val{\hat{X}}| - |P_1| + 1$ be a position in $\valp{\hat{X}^{d}}$ where $P$ occurs,
then $P$ also occurs at $j + c |\val{\hat{X}}|$ in $\valp{\hat{X}^{k}}$ for every positive integer
$c$ with $j + c |\val{\hat{X}}| + |P| - 1 \leq |\valp{\hat{X}^{k}}|$.
Using this observation, 
the index for SLPs can be modified for RLSLPs to achieve the same bounds as in Lemma~\ref{lem:LinearSpaceIndexOfSLP}.
\vspace*{0.5pc}
\noindent \textbf{Index for signature encodings.}
Since signature encodings are RLSLPs,
we can compute $\mathit{Occ}(P,T)$ by querying
$\mathit{report}_{\mathcal{R}}(x_1^{(P,j)}, x_2^{(P,j)},y_1^{(P,j)}, y_2^{(P,j)})$
for ``every'' $1 \leq j < |P|$.
However, the properties of signature encodings allow us to speed up pattern matching
as summarized in the following two ideas:
(1) We can efficiently compute $x_1^{(P,j)}, x_2^{(P,j)}, y_1^{(P,j)}$ and $y_2^{(P,j)}$
using LCE queries in compressed space (Lemma~\ref{lem:ComputePatternRange}).
(2) We can reduce the number of $\mathit{report}_{\mathcal{R}}$ queries from 
$O(|P|)$ to $O(\log |P| \log^* M)$ by using the property of the common sequence of $P$ (Lemma~\ref{lem:pattern_occurrence_lemma1}).
\begin{lemma}\label{lem:ComputePatternRange}
Assume that we have the signature encoding $\mathcal{G}$ of size $w$ for a string $T$ of length $N$, 
$\mathcal{X}$ and $\mathcal{Y}$ of $\mathcal{G}$.
Given a signature $\id{P} \in \mathcal{V}$ for a string $P$ and an integer $j$, 
we can compute $x_1^{(P,j)}, x_2^{(P,j)}, y_1^{(P,j)}$ and $y_2^{(P,j)}$ in $O(\log w (\log N + \log |P| \log^* M))$ time.
\end{lemma}
\begin{proof}
	By Lemma~\ref{lem:sub_operation_lemma}
	we can compute $x_1^{(P,j)}$ and $x_2^{(P,j)}$ 
	on $\mathcal{X}$ by binary search in $O(\log w (\log N + \log |P| \log^* M))$ time.
	Similarly, we can compute $y_1^{(P,j)}$ and $y_2^{(P,j)}$ in the same time. 
	\qed
\end{proof}

\begin{lemma}\label{lem:pattern_occurrence_lemma1}
  Let $P$ be a string with $|P| > 1$.
  If $|\pow{0}{P}| = 1$, then $\mathit{pOcc}_{\mathcal{G}}(P) = \mathit{pOcc}_{\mathcal{G}}(P, 1)$.
  If $|\pow{0}{P}| > 1$, then $\mathit{pOcc}_{\mathcal{G}}(P) = \bigcup_{j \in \mathcal{P}} \mathit{pOcc}_{\mathcal{G}}(P, j)$, where
  $\mathcal{P} = \{ |\valp{u[1..i]}| \mid 1 \leq i < |u|, u[i] \neq u[i+1] \}$ with $u = \uniq{P}$.
\end{lemma}
\begin{proof}
	If $|\pow{0}{P}| = 1$, then $P = a^{|P|}$ for some character $a \in \Sigma$.
	In this case, $P$ must be contained in a node labeled with a signature 
	$e \rightarrow \hat{e}^{d}$ such that $\hat{e} \rightarrow a$ and $d \geq |P|$.
	Hence, all primary occurrences of $P$ can be found by $\mathit{pOcc}_{\mathcal{G}}(P, 1)$.
	
	If $|\pow{0}{P}| > 1$, we consider the common sequence $u$ of $P$.
		Recall that substring $P$ occurring at $j$ in $\val{e}$ is represented by $u$ 
		for any $(e,j) \in \mathit{pOcc}(P)$ by Lemma~\ref{lem:common_sequence2}
	Hence at least $\mathit{pOcc}_{\mathcal{G}}(P) = \bigcup_{i \in \mathcal{P'}} \mathit{pOcc}_{\mathcal{G}}(P,i)$ holds, 
	where $\mathcal{P'} = \{|\valp{u[1]}|, \ldots, |\valp{u[..|u|-1]}|\}$.
	Moreover, we show that $\mathit{pOcc}_{\mathcal{G}}(P,i) = \emptyset$ for any $i \in \mathcal{P'}$ with $u[i] = u[i+1]$.
	Note that $u[i]$ and $u[i+1]$ are encoded into the same signature in the derivation tree of $e$, and
	that the parent of two nodes corresponding to $u[i]$ and $u[i+1]$ has a signature $e'$ in the form $e' \rightarrow u[i]^{d}$.
	Now assume for the sake of contradiction that $e = e'$.
	By the definition of the primary occurrences, $i = 1$ must hold, and hence, $\shrink{0}{P}[1] = u[1] \in \Sigma$.
	This means that $P = u[1]^{|P|}$, which contradicts $|\pow{0}{P}| > 1$.
	Therefore the statement holds.
	\qed
\end{proof}

Using Lemmas~\ref{lem:common_sequence2},~\ref{lem:ComputePatternRange} and~\ref{lem:pattern_occurrence_lemma1},
we get a static index for signature encodings:
\begin{lemma}\label{lem:static_index_lemma}
For a signature encoding $\mathcal{G}$ of size $w$ which generates a text $T$ of length $N$,
there exists a data structure of size $O(w)$ that computes, given a string $P$,
$\mathit{pOcc}_{\mathcal{G}}(P)$ in $O(|P| f_{\mathcal{A}} + \log w \log |P| \log^* M (\log N + \log |P| \log^* M) + q_{w} |\mathit{pOcc}_{\mathcal{S}}(P)|)$ time.
\end{lemma}
\begin{proof}
	We focus on the case $|\pow{0}{P}| > 1$ as the other case is easier to be solved.
	We first compute the common sequence of $P$ in $O(|P|f_{\mathcal{A}})$ time.
	Taking $\mathcal{P}$ in Lemma~\ref{lem:pattern_occurrence_lemma1},
	we recall that $|\mathcal{P}| = O(\log |P| \log^* M)$ by Lemma~\ref{lem:common_sequence2}.
	Then, in light of Lemma~\ref{lem:pattern_occurrence_lemma1}, $\mathit{pOcc}_{\mathcal{G}}(P)$ can be obtained 
	by $|\mathcal{P}| = O(\log |P| \log^* M)$ range reporting queries.
	For each query, we spend $O(\log w (\log N + \log |P| \log^* M))$ time to compute $x_1^{(P,j)}, x_2^{(P,j)}, y_1^{(P,j)}$ and $y_2^{(P,j)}$ by Lemma~\ref{lem:ComputePatternRange}.
	Hence, the total time complexity is 
	\begin{eqnarray}
	O(|P| f_{\mathcal{A}} + \log |P| \log^* M ( \log w (\log N + \log |P| \log^* M) + \hat{q}_{w}) + q_{w} |\mathit{pOcc}_{\mathcal{S}}(P)|) \nonumber \\
	= O(|P| f_{\mathcal{A}} + \log w \log |P| \log^* M (\log N + \log |P| \log^* M) + q_{w} |\mathit{pOcc}_{\mathcal{S}}(P)|). 
	\nonumber
	\end{eqnarray}
	\qed
\end{proof}

In order to dynamize our index of Lemma~\ref{lem:static_index_lemma},
we consider a data structure for ``dynamic'' two-dimensional orthogonal range reporting
that can support the following update operations:
\begin{itemize}
	\item $\mathit{insert}_{\mathcal{R}}(p, x_{\mathit{pred}}, y_{\mathit{pred}})$:
	given a point $p = (x, y)$,
	$x_{\mathit{pred}} = \max \{ x' \in \mathcal{X} \mid x' \leq x \}$ and
	$y_{\mathit{pred}} = \max \{ y' \in \mathcal{Y} \mid y' \leq y \}$,
	insert $p$ to $\mathcal{R}$ and update $\mathcal{X}$ and $\mathcal{Y}$ accordingly.
	\item $\mathit{delete}_{\mathcal{R}}(p)$: given a point 
	$p = (x,y) \in \mathcal{R}$, delete $p$ from $\mathcal{R}$ and update $\mathcal{X}$ and $\mathcal{Y}$ accordingly.
\end{itemize}
We use the following data structure for the dynamic two-dimensional orthogonal range reporting. 
\begin{lemma}[\cite{DBLP:conf/soda/Blelloch08}]\label{lem:RangeQuery}
	There exists a data structure
	that supports $\mathit{report}_{\mathcal{R}}(x_1,x_2,y_1,y_2)$ 
	in $O(\log |\mathcal{R}| + \mathit{occ}(\log |\mathcal{R}| / \log \log |\mathcal{R}|))$ time, 
	and $\mathit{insert}_{\mathcal{R}}(p,i,j)$, 
	$\mathit{delete}_{\mathcal{R}}(p)$ in amortized $O(\log |\mathcal{R}|)$ time, 
	where $\mathit{occ}$ is the number of the elements to output. 
	This structure uses $O(|\mathcal{R}|)$ space.~\footnote{
		The original problem considers a real plane in the paper~\cite{DBLP:conf/soda/Blelloch08}, however, 
		his solution only need to compare any two elements in $\mathcal{R}$ in constant time. 
		Hence his solution can apply to our range reporting problem by maintains $\mathcal{X}$ and $\mathcal{Y}$ 
		using the data structure of order maintenance problem proposed 
		by Dietz and Sleator~\cite{DBLP:conf/stoc/DietzS87}, which enables us to 
		compare any two elements in a list $L$ and insert/delete an element to/from $L$ in constant time.
	}
\end{lemma}

\begin{proof}[Proof of Theorem~\ref{theo:dynamic_index}]
	Our index consists of a dynamic signature encoding $\mathcal{G}$ and a dynamic range reporting data structure of Lemma~\ref{lem:RangeQuery}
	whose $\mathcal{R}$ is maintained as they are defined in the static version.
	We maintain $\mathcal{X}$ and $\mathcal{Y}$ in two ways;
	self-balancing binary search trees for binary search,
	and Dietz and Sleator's data structures for order maintenance.
	Then, primary occurrences of $P$ can be computed as described in Lemma~\ref{lem:static_index_lemma}.
	Adding the $O(\mathit{occ} \log N)$ term for computing all pattern occurrences from primary occurrences,
	we get the time complexity for pattern matching in the statement.

	Concerning the update of our index, 
	we described how to update $\mathcal{G}$ after $\mathit{INSERT}$, $\mathit{INSERT'}$ and $\mathit{DELETE}$ in Lemma~\ref{lem:theorem2}.
	What remains is to show how to update the dynamic range reporting data structure when a signature is added to or deleted from $\mathcal{V}$.
	When a signature $e$ is deleted from $\mathcal{V}$, 
	we first locate $e.{\rm left}^R$ on $\mathcal{X}$ and $e.{\rm right}$ on $\mathcal{Y}$,
	and then execute $\mathit{delete}_{\mathcal{R}}(e.{\rm left}^R, e.{\rm right})$.
	When a signature $e$ is added to $\mathcal{V}$, 
	we first locate
	$x_{\mathit{pred}} = \max \{ x' \in \mathcal{X} \mid x' \leq e.{\rm left}^R \}$ on $\mathcal{X}$ and
	$y_{\mathit{pred}} = \max \{ y' \in \mathcal{Y} \mid y' \leq e.{\rm right} \}$ on $\mathcal{Y}$,
	and then execute $\mathit{insert}_{\mathcal{R}}((e.{\rm left}^R, e.{\rm right}), x_{\mathit{pred}}, y_{\mathit{pred}})$.
	The locating can be done by binary search on $\mathcal{X}$ and $\mathcal{Y}$
	in $O(\log w \log N \log^* M)$ time as Lemma~\ref{lem:ComputePatternRange}.

	Since the number of signatures added to or removed from $\mathcal{G}$
	during a single update operation is upper bounded by Lemma~\ref{lem:up_sig_bound},
	we can get the desired time bounds of Theorem~\ref{theo:dynamic_index}.
	\qed
\end{proof}

%% file: lz.tex
\section{LZ77 factorization in compressed space}\label{sec:appendix_applications}
In this section, we show Theorem~\ref{theo:lzfac}. 
	Note that since each $f_i$ can be represented by the pair $(x_i,|f_i|)$, 
	we compute incrementally $(x_i,|f_i|)$ in our algorithm, 
	where $x_i$ is an occurrence position of $f_i$ in $f_1 \cdots f_{i-1}$.

For integers $j, k$ with $1 \leq j \leq j+k-1 \leq N$,
let $\mathit{Fst}(j,k)$ be the function which returns the minimum integer $i$ 
such that $i < j$ and $T[i..i+k-1] = T[j..j+k-1]$, if it exists. 
Our algorithm is based on the following fact:
\begin{fact}\label{fact:lz77}
Let $f_1, \ldots, f_z$ be the LZ77-factorization of a string $T$. 
Given $f_1, \ldots, f_{i-1}$, we can compute $f_{i}$ with $O(\log |f_i|)$ calls of $\mathit{Fst}(j,k)$
(by doubling the value of $k$, followed by a binary search), 
where $j = |f_1 \cdots f_{i-1}|+1$.
\end{fact}

We explain how to support queries $\mathit{Fst}(j,k)$ using the signature encoding.
We define $e.{\rm min} = \min \mathit{vOcc}(e,S) + |e.{\rm left}|$ for a signature $e \in \mathcal{V}$ with
$e \rightarrow e_{\ell}e_{r}$ or $e \rightarrow \hat{e}^{k}$.
We also define $\mathit{FstOcc}(P,i)$ for a string $P$ and an integer $i$ as follows:
\[
  \mathit{FstOcc}(P,i) = \min \{ e.{\rm min} \mid (e, i) \in \mathit{pOcc}_{\mathcal{G}}(P,i) \}
\]
Then $\mathit{Fst}(j,k)$ can be represented by $\mathit{FstOcc}(P,i)$ as follows:
\begin{eqnarray*}
\mathit{Fst}(j,k) &=& \min \{ \mathit{FstOcc}(T[j..j+k-1],i) - i \mid i \in \{ 1,\ldots, k-1 \} \\
&=& \min \{ \mathit{FstOcc}(T[j..j+k-1],i) - i \mid i \in \mathcal{P} \},
\end{eqnarray*}
where $\mathcal{P}$ is the set of integers in Lemma~\ref{lem:pattern_occurrence_lemma1} with $P = T[j..j+k-1]$.

Recall that in Section~\ref{sec:static_section} 
we considered the two-dimensional orthogonal range reporting problem to enumerate $\mathit{pOcc}_{\mathcal{G}}(P,i)$.
Note that $\mathit{FstOcc}(P,i)$ can be obtained by taking $(e, i) \in \mathit{pOcc}_{\mathcal{G}}(P,i)$ with $e.{\rm min}$ minimum.
In order to compute $\mathit{FstOcc}(P,i)$ efficiently instead of enumerating all elements in $\mathit{pOcc}_{\mathcal{G}}(P,i)$,
we give every point corresponding to $e$ the weight $e.{\rm min}$ and
use the next data structure to compute a point with the minimum weight in a given rectangle.
\begin{lemma}[\cite{DBLP:journals/comgeo/AgarwalAG0Y13}]\label{lem:minimum_weight_report}
Consider $n$ weighted points on a two-dimensional plane.
There exists a data structure which supports the query to return a point with the minimum weight in a given rectangle 
in $O(\log^2 n)$ time, occupies $O(n)$ space, and requires $O(n \log n)$ time to construct.
\end{lemma}

Using Lemma~\ref{lem:minimum_weight_report}, we get the following lemma.
\begin{lemma}\label{lem:prevlemma}
Given a signature encoding $\mathcal{G}$ of size $w$ which generates $T$,
we can construct a data structure of $O(w)$ space in $O(w \log w \log N \log^* M)$ time
to support queries $\mathit{Fst}(j,k)$ in $O(\log w \log k \log^* M (\log N + \log k \log^* M))$ time.
\end{lemma}
\begin{proof}
For construction, we first compute $e.{\rm min}$ in $O(w)$ time using the DAG of $\mathcal{G}$.
Next, we prepare the plane defined by the two ordered sets $\mathcal{X}$ and $\mathcal{Y}$ in Section~\ref{sec:static_section}.
This can be done in $O(w \log w \log N \log^* M)$ time by sorting elements in $\mathcal{X}$ (and $\mathcal{Y}$)
by $\LCEQ$ algorithm (Lemma~\ref{lem:sub_operation_lemma}) and a standard comparison-based sorting.
Finally we build the data structure of Lemma~\ref{lem:minimum_weight_report} in $O(w \log w)$ time.

To support a query $\mathit{Fst}(j,k)$, we first compute $\encpow{\uniq{P}}$
with $P = T[j..j+k-1]$ in $O(\log N + \log k \log^* M)$ time by Lemma~\ref{lem:ComputeShortCommonSequence},
and then get $\mathcal{P}$ in Lemma~\ref{lem:pattern_occurrence_lemma1}.
Since $|\mathcal{P}| = O(\log k \log^* M)$ by Lemma~\ref{lem:common_sequence2},
$\mathit{Fst}(j,k) = \min \{ \mathit{FstOcc}(P,i) - i \mid i \in \mathcal{P} \}$
can be computed by answering $\mathit{FstOcc}$ $O(\log k \log^* M)$ times.
For each computation of $\mathit{FstOcc}(P,i)$,
we spend $O(\log w (\log N + \log k \log^* M))$ time to compute $x_1^{(P,j)}, x_2^{(P,j)}, y_1^{(P,j)}$ and $y_2^{(P,j)}$ by Lemma~\ref{lem:ComputePatternRange},
and $O(\log^2 w)$ time to compute a point with the minimum weight in the rectangle $(x_1^{(P,j)}, x_2^{(P,j)}, y_1^{(P,j)}, y_2^{(P,j)})$.
Hence it takes $O(\log k \log^* M (\log w (\log N + \log k \log^* M) + \log^2 w)) = O(\log w \log k \log^* M (\log N + \log k \log^* M))$ time in total.
\qed
\end{proof}

We are ready to prove Theorem~\ref{theo:lzfac} holds.
\begin{proof}[Proof of Theorem~\ref{theo:lzfac}]
We compute the $z$ factors of the LZ77-factorization of $T$ incrementally 
by using Fact~\ref{fact:lz77} and Lemma~\ref{lem:prevlemma} in $O(z \log w \log^3 N (\log^* M)^2)$ time.
Therefore the statement holds.
\qed
\end{proof}
We remark that we can similarly compute the Lempel-Ziv77 factorization \emph{with} self-reference of a text (defined below) in the same time and same working space. 

\begin{definition}[Lempel-Ziv77 factorization with self-reference~\cite{LZ77}]
The Lempel-Ziv77 (LZ77) factorization of a string $s$ with self-references is
a sequence $f_1, \ldots, f_k$ of non-empty substrings of $s$ 
such that $s = f_1 \cdots f_k$,
$f_1 = s[1]$, 
and for $1 < i \leq k$, if the character $s[|f_1..f_{i-1}|+1]$ does not occur in $s[|f_1..f_{i-1}|]$, then $f_i = s[|f_1..f_{i-1}|+1]$, otherwise $f_i$ is the longest 
prefix of $f_i \cdots f_k$ which occurs at some position $p$, where $1 \leq p \leq |f_1 \cdots f_{i-1}|$.
\end{definition}

%% file: examples_figs.tex
\section{Appendix: Supplementary Examples and Figures}\label{sec:Example_Section}
\begin{example}[$\encblockd{p}$ and $\encpow{s}$] \label{ex:Encblock}
Let $\log^* W = 2$, and then $\deltaLR{L} = 8, \deltaLR{R} = 4$.\\
If $p = 1,2,3,2,5,7,6,4,3,4,3,4,1,2,3,4,5$ and $d = 1,0,0,1,0,1,0,0,1,0,0,0,1,0,1,0,0$,
then $\encblockd{p} = (1,2,3),(2,5),(7,6,4),(3,4,3,4),(1,2),(3,4,5)$, $|\encblockd{p}| = 6$ and $\encblockd{p}[2] = (2, 5)$. 
For string $s = aaaabbbbbabbaa$,
$\encpow{s} = a^4b^5a^1b^2a^2$ and
$|\encpow{s}| = 5$ and $\encpow{s}[2] = b^5$.
\end{example}

\begin{example}[SLP]\label{ex:SLP}
	Let $\mathcal{S} = (\Sigma, \mathcal{V}, \mathcal{D}, S)$ be the SLP
	s.t.
	$\Sigma = \{A, B, C \}$, $\mathcal{V} = \{ X_1, \cdots , X_{11} \}$, 
	$\mathcal{D} = \{ 
	X_{1} \rightarrow A, X_{2} \rightarrow B, X_{3} \rightarrow C, 
	X_4 \rightarrow X_{3}X_{1}, X_5 \rightarrow X_{4}X_{2}, 
	X_6 \rightarrow X_{5}X_{5}, X_7 \rightarrow X_{2}X_{3}, 
	X_8 \rightarrow X_{1}X_{2}, X_9 \rightarrow X_{7}X_{8}, 
	X_{10} \rightarrow X_{6}X_{9}, X_{11} \rightarrow X_{10}X_{6}
	\}$, $S = X_{11}$, 
	the derivation tree of $S$ represents $CABCABBCABCABCAB$.
\end{example}
\begin{example}[RLSLP]\label{ex:tree}
Let $\mathcal{G} = (\Sigma, \mathcal{V}, \mathcal{D}, S)$ be an RLSLP, 
where $\Sigma = \{A, B\}$, $\mathcal{V} = \{1, \ldots , 24 \}$, $\mathcal{D} = \{ 
1 \rightarrow A, 2 \rightarrow B, 
3 \rightarrow 1^1, 4 \rightarrow 2^1, 5 \rightarrow 2^2, 6 \rightarrow 1^2,
7 \rightarrow (3,4), 8 \rightarrow (3,5), 9 \rightarrow (8,3), 10 \rightarrow (4,3), 
11 \rightarrow (10,4), 12 \rightarrow (11,6), 
13 \rightarrow 7^3, 14 \rightarrow 9^1, 15 \rightarrow 10^1, 16 \rightarrow 12^1, 17 \rightarrow 10^3, 
18 \rightarrow (13,14), 19 \rightarrow (18,15), 20 \rightarrow (16,17), 
21 \rightarrow 19^1, 22 \rightarrow 20^1, 
23 \rightarrow (21,22), 
24 \rightarrow 23^1
\}$, and $S = 24$.
The derivation tree of the start symbol $S$ represents a single string 
$T = ABABABABBABABABAABABABA$. 
\end{example}

\begin{example}[Signature encoding]\label{ex:signature_dictionary}
Let $\mathcal{G} = (\Sigma, \mathcal{V}, \mathcal{D}, S)$ be an RLSLP of Example~\ref{ex:tree}. 
Assuming $\encblock{\pow{0}{T}} = (3,4)^3,(3,5,3),(4,3),(4,3,4,6), (4,3)^3$,  
$\encblock{\pow{2}{T}} = (13,14,15), (16,17)$ and 
$\encblock{\pow{3}{T}} = (21,22)$ hold, 
$\mathcal{G}$ is the signature encoding of $T$ and $\id{T} =  24$. 
Here, $\sig{(13,14)} = 18$, $\sig{(4,3,4,6)} = 12$, $\sig{(4,5)} = \rm{undefined}$.
See Fig.~\ref{fig:SignatureTree} for an illustration of the derivation tree of $\mathcal{G}$ and the corresponding DAG.
\end{example}

\begin{example}[Primary occurrences]\label{ex:primary_occurrence}
	Let $\mathcal{S}$ be the SLP of Example~\ref{ex:SLP}. 
	Given a pattern $P = BCAB$, 
	then $P$ occurs at $3$, $7$, $10$ and $13$ in the string $T$ represented by SLP $\mathcal{S}$. 
	Hence $\mathit{Occ}(P,T) = \{3,7,10,13\}$. 
	On the other hand, $P$ occurs at $3$ in $\val{X_6}$ and $P$ is divided by $X_{5}$ and $X_{5}$, where $X_{6} \rightarrow X_{5}X_{5}$.
	Similarly, divided $P$ occurs at $1$ in $\val{X_9}$, at $10$ in $\val{X_{11}}$. 
	Hence $\mathit{pOcc}_{\mathcal{S}}(P) = \{ (X_{6},3), (X_{11}, 10), (X_{9},1)\}$. 
	Specifically
	$\mathit{pOcc}_{\mathcal{S}}(P,1) = \{ (X_{6},3), (X_{11}, 10 )\}$, 
	$\mathit{pOcc}_{\mathcal{S}}(P,2) = \{ (X_{9},1) \}$ and 
	$\mathit{pOcc}_{\mathcal{S}}(P,3) = \phi$.  
	Hence we can also compute $\mathit{Occ}(P,T) = \{3,7,10,13\}$ by 
	$\mathit{vOcc}(X_{6},S) = \{1,11\}$, 
	$\mathit{vOcc}(X_{9},S) = \{7\}$, and 
	$\mathit{vOcc}(X_{11},S) = \{1\}$. 
	See also Fig.~\ref{fig:grid}. 
\end{example}

\begin{figure}[ht]
\begin{center}
  \includegraphics[scale=0.7]{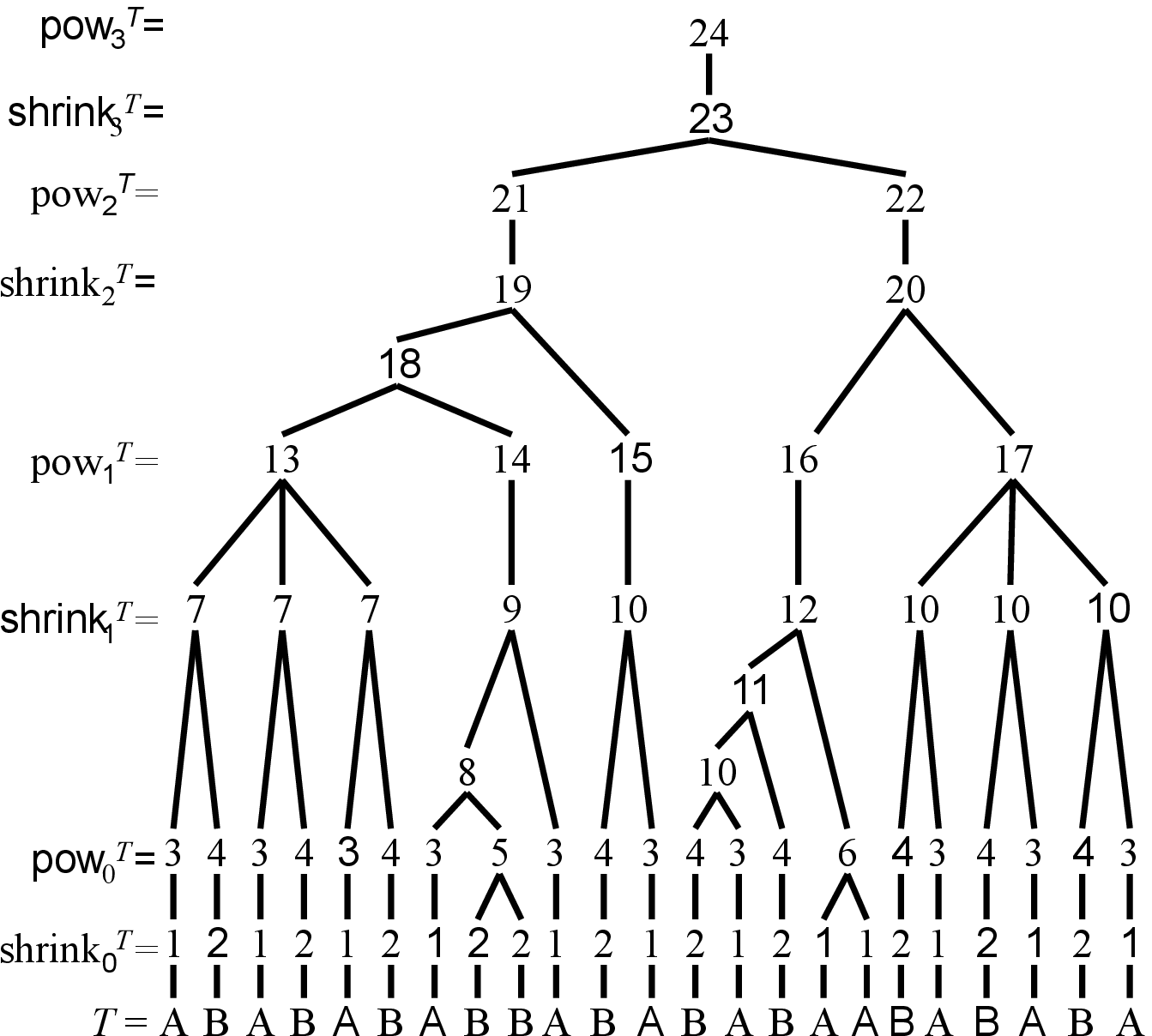}  
  \includegraphics[scale=0.8]{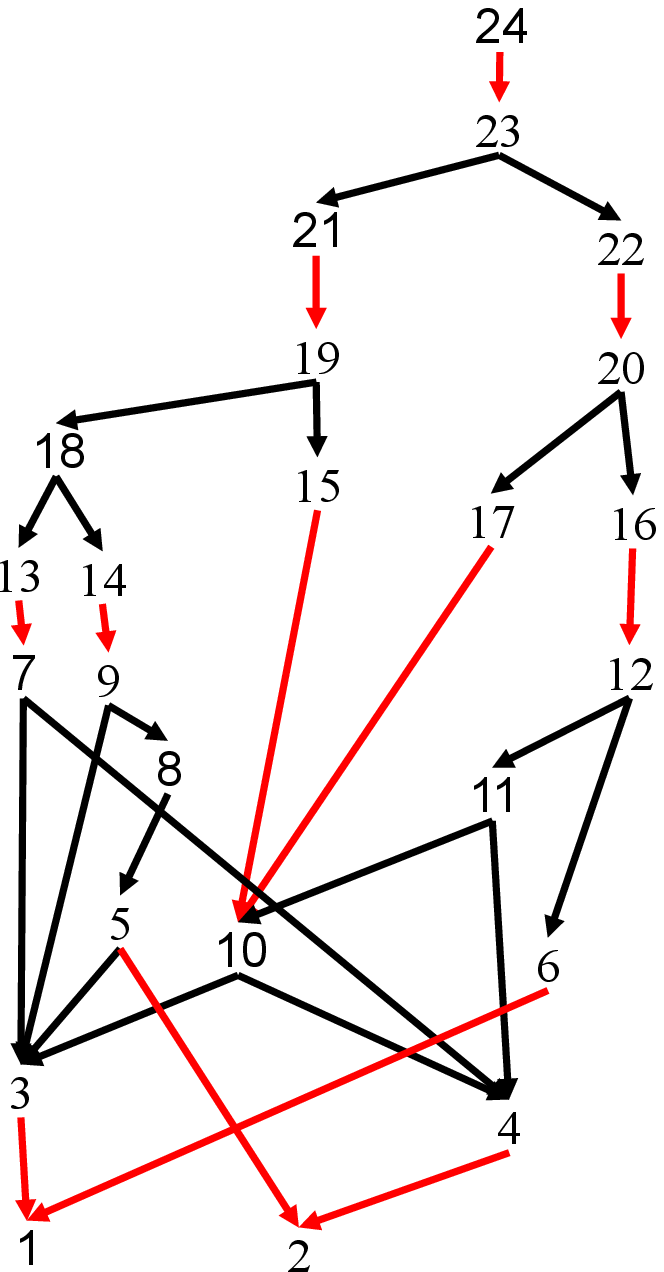}
  \caption{
  The derivation tree of $S$ (left) and the DAG for $\mathcal{G}$ (right) of Example~\ref{ex:tree}. 
  In the DAG, the black and red arrows represent $e \rightarrow e_{\ell}e_r$ and
  $e \rightarrow \hat{e}^{k}$ respectively.   
  In Example~\ref{ex:signature_dictionary}, $T$ is encoded by signature encoding.
  }
  \label{fig:SignatureTree}
\end{center}
\end{figure}

\begin{figure}[h]
\begin{center}
  \includegraphics[scale=0.5]{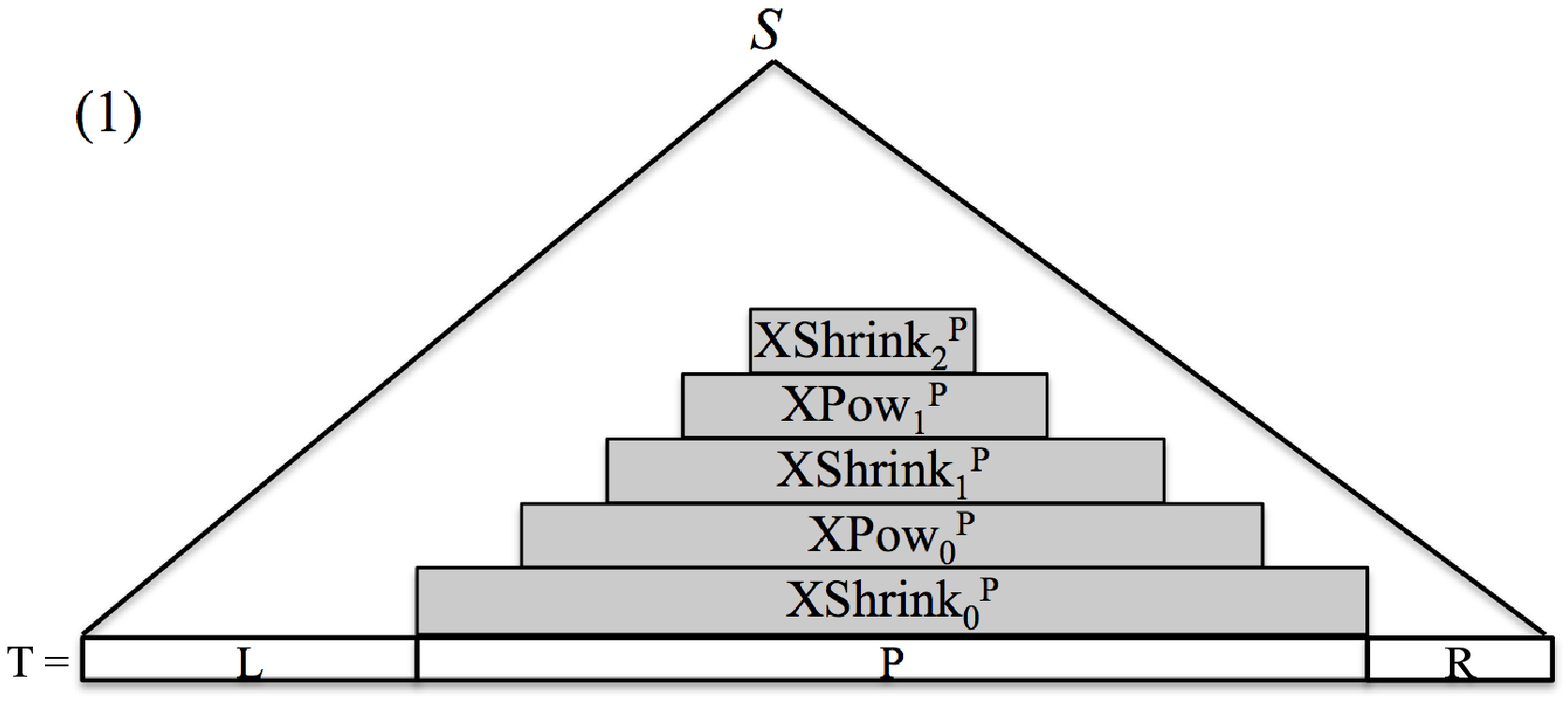}
  \includegraphics[scale=0.6]{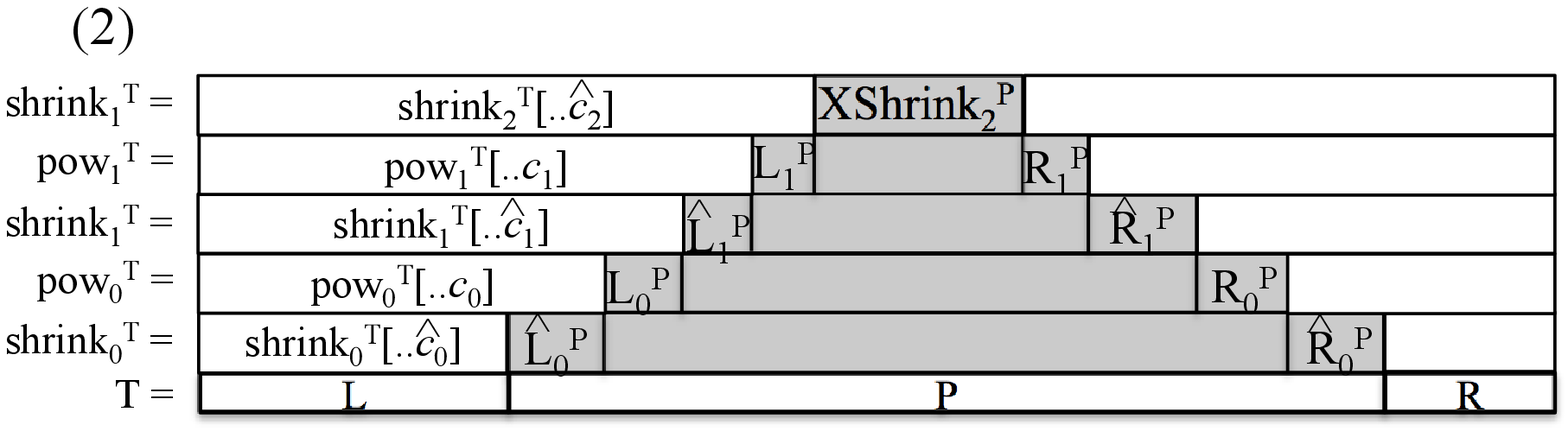}
  \includegraphics[scale=0.5]{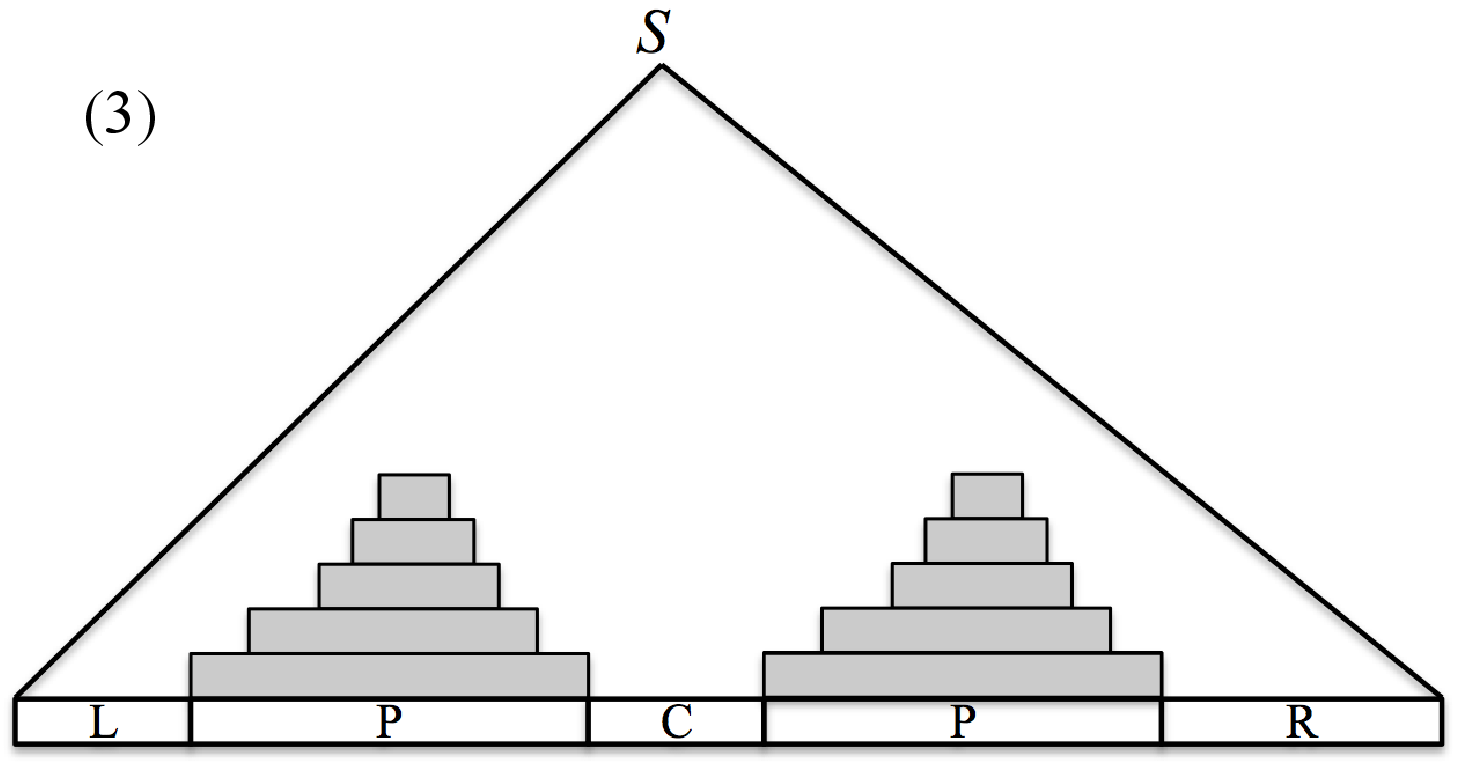}
  \caption{
  Abstract images of consistent signatures of substring $P$ of text $T$,
  on the derivation trees of the signature encoding of $T$.
  Gray rectangles in Figures (1)-(3) represent common signatures for occurrences of $P$. 
  (1) Each $\mathit{XShrink}_{t}^{P}$ and $\mathit{XPow}_{t}^{P}$ occur on substring $P$ 
  in $\mathit{shrink}_{t}^{T}$ and $\mathit{Pow}_{t}^{T}$, respectively, where $T = LPR$.
  (2) The substring $P$ can be represented by 
  $\hat{L}_{0}^{P}L_{0}^{P}\hat{L}_{1}^{P}L_{1}^{P}\mathit{XShrink}_{2}^{P}R_{1}^{P}\hat{R}_{1}^{P}R_{0}^{P}\hat{R}_{0}^{P}$. 
  (3) There exist common signatures on every substring $P$ in the derivation tree.
  } 
  \label{fig:CommonSequence}
\end{center}
\end{figure}

\begin{example}[SLP]\label{ex:SLPGrid}
	Let $\mathcal{S}$ be the SLP of Example~\ref{ex:SLP}. 
	Then, 
	\begin{eqnarray*}
		\mathcal{X} &=& \{ x_1, x_4, x_2, x_8, x_5, x_9, x_6, x_{10}, x_{11}, x_3, x_7 \}, \\
		\mathcal{Y} &=& \{ y_1, y_8, y_2, y_7, y_9, y_3, y_4, y_{5}, y_{6}, y_{10}, y_{11} \}, \\
	\end{eqnarray*}
	$x_i = \val{X_i}^{R}$, $y_i = \val{X_i}$ for any $X_i \in \mathcal{V}$.
	See also Fig.~\ref{fig:grid}. 
\end{example}
\begin{figure}[ht]
	\begin{center}
		\includegraphics[scale=0.7]{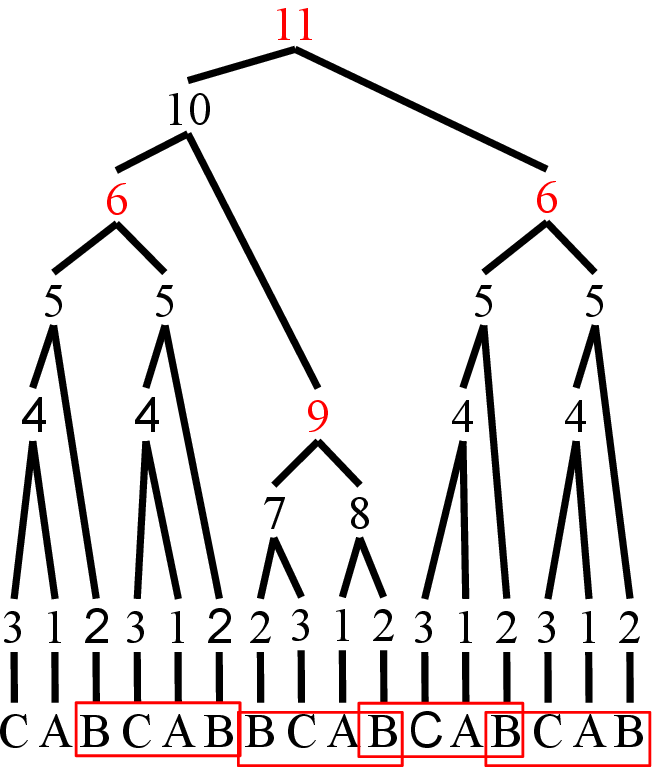}
		\includegraphics[scale=0.48]{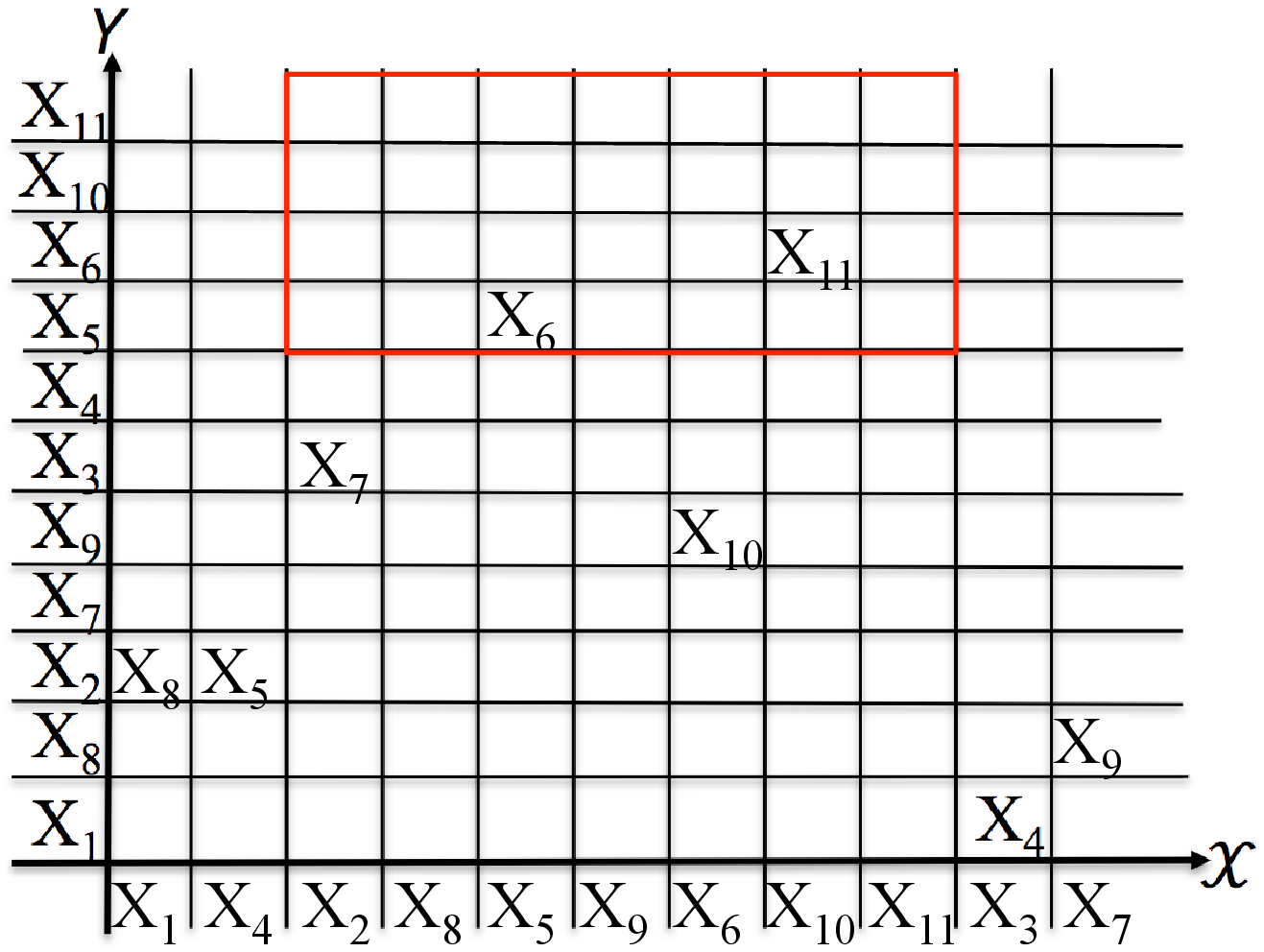}
		\caption{
			The left figure is the derivation tree of SLP $\mathcal{S}$ of Example~\ref{ex:SLP}, which derives the string $T$. 
			The red rectangles on $T$ represent all occurrences of $P = BCAB$ in $T$. 
			The right grid represents the relation between $\mathcal{X}$, $\mathcal{Y}$ and $\mathcal{R}$ of Example~\ref{ex:SLPGrid}. 
			The red rectangle on the grid is a query rectangle $(x^{(P,1)}_1,x^{(P,1)}_2,y^{(P,1)}_1,y^{(P,1)}_2)$, where 
			$x^{(P,1)}_1 = x_2$, $x^{(P,1)}_2 = x_{11}$, $y^{(P,1)}_1 = y_5$ and $y^{(P,1)}_2 = y_{11}$. 
			Therefore, $\mathit{report}_{\mathcal{R}}(x^{(P,1)}_1,x^{(P,1)}_2,y^{(P,1)}_1,y^{(P,1)}_2) = \{ X_{6}, X_{11} \}$.
		}
		\label{fig:grid}
	\end{center}
\end{figure}